\documentclass[journal,10pt]{IEEEtran}
\usepackage{blindtext}
\usepackage{graphicx}
\usepackage{amsmath}
\usepackage{mathrsfs}
\usepackage{amsthm}
\usepackage{multirow}
\usepackage{amsmath}
\usepackage{mathtools}

\DeclareMathOperator*{\argmin}{argmin}
\usepackage{yhmath}
\usepackage{amssymb}
\usepackage{array}
\usepackage{leftidx}
\usepackage{longtable}
\usepackage{stfloats}
\usepackage{mathalpha}
\usepackage{cite}
\usepackage{url}
\usepackage{subfigure}
\usepackage{comment}
\usepackage{color}
\usepackage{algorithm}
\usepackage{algorithmic}
\usepackage{lipsum}

\begin{document}
\theoremstyle{plain}
\newtheorem{thm}{Theorem}
\newtheorem{remark}{Remark}
\newtheorem{lemma}{Lemma}
\newtheorem{prop}{Proposition}
\newtheorem{cor}{Corollary}
\theoremstyle{definition}
\newtheorem{defn}{Definition}
\newtheorem{condi}{Condition}
\newtheorem{assump}{Assumption}

\title{Newton Methods in Generalized Nash Equilibrium Problems with Applications to Game-Theoretic Model Predictive Control}

\author{Mushuang Liu,~\IEEEmembership{Member,~IEEE,} and Ilya Kolmanovsky,~\IEEEmembership{Fellow,~IEEE}

\thanks{M. Liu is with the Department of Mechanical Engineering, Virginia Tech, Blacksburg, VA, USA (mushuang@vt.edu).
}

\thanks{I. Kolmanovsky is with the Department of  Aerospace Engineering, University of Michigan, Ann Arbor, MI, USA (ilya@umich.edu).
}
\thanks{This work was supported by DARPA Young Faculty Award with the grant number D24AP00321.}
}
\maketitle

\begin{abstract}
We prove input-to-state stability (ISS) of perturbed Newton-type methods for generalized equations arising from Nash equilibrium (NE) and generalized NE (GNE) problems. This ISS property allows the use of inexact computations in equilibrium-seeking to enable fast solution tracking in dynamic systems such as in model predictive control (MPC). For NE problems, we address the local convergence of perturbed Josephy-Newton methods from the variational inequality (VI) stability analysis, and establish the ISS result under less restrictive regularity conditions compared to the existing results established for nonlinear optimization. Agent-distributed algorithms are also developed. For GNE problems, since they cannot be reduced to VI problems in general, we use semismooth Newton methods to solve the semismooth equations arising from the Karush-Kuhn-Tucker (KKT) systems of the GNE problem and establish the ISS result under a quasi-regularity condition. To illustrate the use of the ISS in dynamic systems, applications to constrained game-theoretic MPC (CG-MPC) are studied with time-distributed solution-tracking for real-time implementation. Boundness of tracking errors is proven. Numerical examples are reported. 
\end{abstract}


\section{Introduction} 
A Nash equilibrium (NE) problem \cite{nash1950equilibrium} concerns solving interdependent optimization problems subject to independent agents' constraints:
\begin{subequations}\label{NE_problem}
    \begin{align}
        \min_{a_{i}}\quad & J_i(a_i,a_{-i})\\
         s.t.\quad & a_{i}\in\mathcal{A}_i,
    \end{align}
\end{subequations}
where $i\in\{1,2,...,N\}$ represents the agent, $\mathcal{A}_i\subseteq\mathbb{R}^{n_i}$ is the feasible set of agent $i$, and $J_i:\mathbb{R}^{n}\rightarrow\mathbb{R}$ is the cost function of agent $i$. When agents have coupled constraints, a generalized NE (GNE) problem \cite{GNEP} is composed as:
\begin{subequations}\label{GNE_problem}
    \begin{align}
        \min_{a_{i}}\quad & J_i(a_i,a_{-i})\\
         s.t.\quad & a_{i}\in\mathcal{A}_i(a_{-i}),
    \end{align}
\end{subequations}
where $\mathcal{A}_i$ depends on $a_{-i}$, leading to coupled (rather than independent) agents' constraints. 

Similar to nonlinear optimization problems, the stationary conditions for the NE and GNE problems can be expressed as a generalized equation using normal cone mapping \cite{GNEP_KKT}:
\begin{equation}\label{eq:GE1}
    f(z)+\mathcal{F}(z)\ni 0,
\end{equation}
where $f$ is a function, $\mathcal{F}$ is a set-valued mapping acting between Banach spaces, and $z$ could be either the decision variable $a$ or the primal-dual variable $(a,\lambda)$ from the Karush-Kuhn-Tucker (KKT) conditions \cite{GNEP_KKT} depending on the setting of the problem. One way of solving the generalized equation \eqref{eq:GE1} is to use Newton-type methods \cite{book_pang}, e.g., Josephy-Newton method \cite{josephy_1} if $f$ is continuously Fréchet differentiable \cite{book_pang}. The Josephy-Newton method extends the classical Newton's method for solving equations $f(z)=0$ to more general variational problems \eqref{eq:GE1} with the  iterations of the form: 
\begin{equation}\label{eq:JN1}
    f(z^k)+\nabla f(z^k)(z^{k+1}-z^k)+\mathcal{F}(z^{k+1})\ni 0,
\end{equation}
where $k=0,1,...$ is the number of iterations with $z^0$ being a given starting point, and $\nabla f$ is the gradient of $f$. 
Local convergence of this Josephy-Newton method for nonlinear optimization problems has been established in the literature \cite{dontchev2013convergence}, and most of them are based on metric regularity conditions requiring that $f+\mathcal{F}$ is strongly regular or strongly subregular \cite{dontchev2013convergence,cunis2024input,dontchev2009implicit,dontchev2021lectures}. We show in this paper that similar properties can be established for generalized equations arising from NE problems \eqref{NE_problem} and that the local convergence can be achieved under relaxed regularity conditions (specifically, regularity but not strong regularity) via variational inequality (VI) stability analysis. Agent-distributed Josephy-Newton methods where each agent performs local computations for its own decision variable updates are also developed.

Perturbed Newton-type methods \cite{cunis2024input}, where the perturbations may come from an inexact evaluation of the gradient $\nabla f$ or a nonzero reminder in solving \eqref{eq:JN1}, play a critical role when analyzing the interconnection between the optimizer and a dynamic system the controller of which requires solving consecutive optimization problems such as in model predictive control (MPC) \cite{liao2020time}.  For these systems, the input-to-state stability (ISS) of the perturbed Newton method is the key property to ensure trackable optimizer performance and closed-loop system properties such as stability, robustness, and constraint satisfaction. In a recent study \cite{cunis2024input}, the ISS of perturbed Newton methods for nonlinear optimization was established under the strong regularity assumption. In this paper, we extend such results to the NE and GNE problems and show that the ISS holds under relaxed regularity conditions.

Note that for a GNE problem, its first-order necessary conditions do not reduce to a VI problem  due to the coupled agents' constraints; instead, they lead to a quasi-VI (QVI) problem \cite{GNEP}, the solution to which is not well understood yet. By using complementarity functions, the stationary conditions from coupled KKT systems can  be reformulated as a system of nonsmooth equations, which could be solved by semismooth Newton methods \cite{GNEP_newton}. However,  because the conventional VI stability analysis is not applicable (unless a variational-GNE (v-GNE) problem), it has been an open question whether the ISS holds for the perturbed semismooth Newton methods arising from the GNE problems. We answer this question by proving the ISS under suitable nonsingularity conditions with proven $Q$-quadratic local convergence. 

To illustrate the role of the ISS in closed-loop dynamic systems, we study constrained game-theoretic MPC (CG-MPC) problems \cite{CGMPC_application2}. The CG-MPC concerns solving consecutive NE or GNE problems along the state trajectory, which is computationally demanding in general. One way to reduce the computational load is to employ the time-distributed optimization (TDO) strategy by distributing the optimizer iterations over time to track approximate solutions of the optimal control problem at each time instant \cite{liao2020time}. With such an implementation strategy, the solver itself becomes a dynamic system evolving in parallel with the plant, resulting in an interconnected plant-optimizer feedback system.  We show in this paper that the appealing properties of the TDO-based MPC, including the closed-loop stability and robustness \cite{liao2020time}, also hold in the CG-MPC problem with the Newton-type methods developed in this paper.


\textit{Notation.} We use $\mathbb{R}$, $\mathbb{R}_{+}$, and  $\mathbb{R}_{++}$ to denote the set of real, non-negative real, and positive real numbers respectively, $\mathbb{N}_{+}$ and $\mathbb{N}_{++}$ to denote the set of non-negative and positive integers respectively, $\|\cdot\|$ to denote the Euclidean norm, and $\|\cdot\|_{\infty}$ to denote the infinity norm.

\section{Newton Methods for NE Problems}\label{sec:NE}

Consider the NE problem \eqref{NE_problem}.
 We use $\mathcal{G}(\mathcal{A},J)$ to represent the game, where $\mathcal{A}=\prod_{i=1}^N \mathcal{A}_i$ and $J=\{J_i\}_{i=1}^N$. Let $a=\{a_i,a_{-i}\}$ represent the set of all players' decision variables with $-i$ being the set of all other agents except for agent $i$. Note that although \eqref{NE_problem} involves constrained optimization, it is a NE (rather than a GNE) problem because $\mathcal{A}_i$ does not depend on $a_{-i}$. 
If $a^*=\{a_1^*,...,a_N^*\}$ solves \eqref{NE_problem} for all $i$, i.e.,
\begin{equation}\label{Op_agent}
a_i^*\in\argmin_{a_{i}\in\mathcal{A}_i}J_i(a_i,a_{-i}^*)
\end{equation}
holds for all $i\in\mathcal{N}$, then $a^*$ is a NE.  

We make the following assumptions throughout this section.

A1. (Smoothness): For every $i$, the cost function $J_i$ is twice continuously differentiable in its arguments, and its second-order derivatives are Lipschitz continuous;

A2. (Player-wise convexity): For every $i$  and every $a_{-i}$, the cost function $J_i(\cdot,a_{-i})$ is convex, and the set $\mathcal{A}_i$ is compact and convex.

Under A1 and A2, the problem \eqref{NE_problem} is guranteed to be solvable \cite{NE_existence}, i.e., a NE always exists.  The following theorem shows that under A1 and A2, the NE problem \eqref{NE_problem} is equivalent to a VI problem.

\begin{thm}[Corollary 3.4 in \cite{GNEP}]\label{Equivalence_VI}
Consider the NE problem \eqref{NE_problem}. Let A1 and A2 hold. A point $a^*\in\mathcal{A}$ is a NE of \eqref{NE_problem} if and only if it is a solution of the  VI$(\mathcal{A},F)$, where $F:\mathbb{R}^{n}\rightarrow\mathbb{R}^n$ is the pseudogradient of the game defined by
\begin{equation}\label{P_gradient}
   F(a)=\{\nabla_{a_i}J_i(a_i,a_{-i})\}_{i=1}^{N},
\end{equation}
and VI$(\mathcal{A},F)$ is the variational inequality. 
\end{thm}
By Theorem  \ref{Equivalence_VI}, $a^*$ is a NE if and only if it solves the following generalized equation:
\begin{equation}\label{eq:GE}
    F(a)+\mathcal{N}_{\mathcal{A}}(a)\ni 0,
\end{equation}
where $\mathcal{N}_{\mathcal{A}}(a)=\{b\in\mathbb{R}^n:b^T(a'-a)\leq 0, \forall a'\in\mathcal{A}\}$ is the normal cone to $\mathcal{A}$ at $a$.

Denote the \textit{game Hessian} $H:\mathbb{R}^{n}\rightarrow\mathbb{R}^{n\times n}$ for \eqref{NE_problem} as \begin{equation}\label{Hessian}
    H(a)= \begin{bmatrix}
           \nabla^2_{a_1a_1}J_1(a) &  \nabla^2_{a_1a_2}J_1(a) & \cdots &\nabla^2_{a_1a_N}J_1(a)\\
           \nabla^2_{a_2a_1}J_2(a) &  \nabla^2_{a_2a_2}J_2(a) & \cdots &\nabla^2_{a_2a_N}J_2(a) \\
           \vdots & \vdots &\vdots &\vdots \\
           \nabla^2_{a_Na_1}J_N(a) &  \nabla^2_{a_Na_2}J_N(a) & \cdots &\nabla^2_{a_Na_N}J_N(a)
         \end{bmatrix}.
\end{equation} 
\begin{remark}
    Different from that in constrained optimization, the pseudogradient $F$ in VI$(\mathcal{A},F)$ is not integrable on $\mathcal{A}$  (see \cite[Pages 13-14]{book_pang} for discussions on integrability) because the game Hessian \eqref{Hessian} is non-symmetric in general  \cite[Theorem 1.3.1]{book_pang}, unless the game is a potential game. 
\end{remark}
Denote the critical cone of agent $i$'s optimization problem when $a_{-i}=a_{-i}^*$ (i.e., problem \eqref{Op_agent}) as
\begin{equation} 
\mathcal{C}_i(a_i^*;\mathcal{A}_i,\nabla_{a_i}J_i(\cdot,a_{-i}^*))=\mathcal{T}(\mathcal{A}_i;a_i^*)\cap \nabla_{a_i}J_i(a^*)^{\perp},
\end{equation}
where $\mathcal{T}(\mathcal{A}_i;a_i^*)$ denotes the tangent cone of the set $\mathcal{A}_i$ at a point $a_i^*\in\mathcal{A}_i$, and $\nabla_{a_i}J_i(a^*)^{\perp}$ denotes the orthogonal complement of the gradient $\nabla_{a_i}J_i(a^*)$.
Denote $\mathcal{C}(a^*;\mathcal{A},F)=\prod_{i=1}^{N} \mathcal{C}_i(a_i^*;\mathcal{A}_i,\nabla_{a_i}J_i(a^*))$ as the critical cone of the game.

The game Hessian $H(a^*)$ is strictly semicopositive on the cone $\mathcal{C}(a^*;\mathcal{A},F)$ if for every nonzero vector $a\in\mathcal{C}(a^*;\mathcal{A},F)$,
\begin{equation}
    \max_{1\leq i\leq N} \quad a_i^T\sum_{j=1}^N H_{ij}a_j>0,
\end{equation}
where $H_{ij}=\nabla^2_{a_ia_j}J_i(a^*)$.
This strictly semicopositive property is important in the NE local uniqueness and stability analysis as discussed in the next subsection. 

\subsection{Stability of NE}
In terms of the stability of an NE $a^*$ of the game $\mathcal{G}(\mathcal{A},J)$, we are interested in (i) when $a^*$ is an isolated (i.e., locally unique) NE, and (ii) whether the ``nearby games" of $\mathcal{G}(\mathcal{A},J)$ always have a NE that is ``close" to $a^*$. 

To start with, let us first define the isolatedness of a NE $a^*$. This property is desirable for showing that this $a^*$ is indeed an ``attractor" of all NEs of ``nearby games", and it is central to the local convergence analysis of Newton-type methods. 
\begin{defn}
    A NE $a^*$ of the game $\mathcal{G}(\mathcal{A},J)$ is said to be locally unique, or isolated, if there exists a neighborhood $\Omega$ of $a^*$ such that 
    \begin{equation}\label{eq:uniquess}
        \text{NESet}(\mathcal{G})\cap \Omega=\{a^*\},
    \end{equation}
    where $\text{NESet}(\mathcal{G})$ is the set of NEs of the game $\mathcal{G}(\mathcal{A},J)$.
\end{defn} 

The next theorem establishes the conditions under which $a^*$ is an isolated NE of the game $\mathcal{G}(\mathcal{A},J)$. 
\begin{thm}[Proposition 12.14 in \cite{pang_bookchapter}]\label{uniqueness}
    Let A1 and A2 hold. Let $a^*$ be a NE of the game \eqref{NE_problem}. If $H(a^*)$ is strictly semicopositive on the critical cone $\mathcal{C}(a^*;\mathcal{A},F)$, then $a^*$ is an isolated NE. 
\end{thm}

Next, to be able to define NE stability, we first  formally define ``nearby games" of a given game $\mathcal{G}(\mathcal{A},J)$ restricted to a subset of $\mathcal{A}$, by introducing a \textit{game neighborhood} concept. 
\begin{defn}\label{game_neighbor}
    Given a game $\mathcal{G}(\mathcal{A},J)$ satisfying A1 and A2, a game $\Tilde{\mathcal{G}}(\mathcal{A},\Tilde{J})$ is said to be in the $\epsilon$-neighborhood of $\mathcal{G}(\mathcal{A},J)$ restricted to $\mathcal{S}\subset\mathcal{A}$, denoted by $\Tilde{\mathcal{G}}\in\pmb{\mathrm{B}}(\mathcal{G};\epsilon,\mathcal{S})$, if $\Tilde{\mathcal{G}}$ satisfies A1 and A2 and their pseudogradients satisfy 
    \begin{equation}
        \|F-\Tilde{F}\|_{\mathcal{S}}\coloneqq \sup_{a\in\mathcal{S}}\|F(a)-\Tilde{F}(a)\|<\epsilon,
    \end{equation}
    where $\Tilde{F}$ is the pseudogradient of the game $\Tilde{\mathcal{G}}$. 
\end{defn}

 
\begin{defn}\label{DE:NE_stability}
    A NE $a^*$ of the game $\mathcal{G}(\mathcal{A},J)$ is said to be \textit{stable} if for every open neighborhood $\Omega$ of $a^*$ satisfying \eqref{eq:uniquess} and for every $\Tilde{\mathcal{G}}\in\pmb{\mathrm{B}}(\mathcal{G};\epsilon,\mathcal{A}\cap \mathsf{cl}\Omega)$, where $\mathsf{cl}\Omega$ is the closure of $\Omega$,
    \begin{equation}\label{eq:local_solvability}
        \text{NESet}(\Tilde{\mathcal{G}}) \cap \mathsf{cl}\Omega\neq \emptyset;
    \end{equation}
    and, in addition, there exists two positive scalars $c$ and $\epsilon$ such that, for every $\Tilde{\mathcal{G}}\in\pmb{\mathrm{B}}(\mathcal{G};\epsilon,\mathcal{A}\cap \mathsf{cl}\Omega)$ and every $\Tilde{a}\in\text{NESet}(\Tilde{\mathcal{G}}) \cap \mathsf{cl}\Omega$,
    \begin{equation}\label{eq:error_def}
        \|\Tilde{a}-a^*\|\leq c\|e(\Tilde{a})\|,
    \end{equation}
    where $e(a)=F(a)-\Tilde{F}(a)$ is the difference function of the pseudogradients of the two games.
\end{defn}

Definition \ref{DE:NE_stability} states that a NE of a game is stable if (i) its nearby games are locally solvable (i.e., \eqref{eq:local_solvability}), and (ii) the distance between the NE of the given game and that of a nearby game yields an upper bound in terms of the maximum deviation of the pseudogradients of the two games (i.e., \eqref{eq:error_def}). This stability concept corresponds to the regularity condition in VI problems \cite[Definition 5.3.2]{book_pang}. 
The following theorem establishes the stability result. The proofs of these and other results are in the Appendix.

\begin{thm}\label{Th:stability}
Let $a^*$ be a NE of the game $\mathcal{G}(\mathcal{A},J)$. If A1 and A2 hold and if $H(a^*)$ is strictly semicopositive on the critical cone $\mathcal{C}(a^*;\mathcal{A},F)$, then $a^*$ is a stable NE.
\end{thm}

\subsection{Josephy-Newton method}\label{sec:JN}
We consider the Josephy-Newton method to solve the generalized equation  \eqref{eq:GE}. Specifically, given $a^k$, one solves the following semi-linearized generalized equation to find $a^{k+1}$:
\begin{equation}\label{eq:JN}
    F(a^k)+H(a^k)(a-a^k)+\mathcal{N}_{\mathcal{A}}(a)\ni 0,
\end{equation}
where $H(a^k)$ is the game Hessian \eqref{Hessian} evaluated at $a^k$, and $F(a^k)+H(a^k)(a-a^k)$ is the linearization of the pseudogradient $F$ at $a^k$. 


The next theorem shows that the sequence $\{a^k\}$, generated by \eqref{eq:JN}, converges $Q$-quadratically to $a^*$. 
\begin{thm}\label{tm:stability_NE}
    Let $a^*$ be a stable NE of the game $\mathcal{G}(\mathcal{A},J)$ and $F$ be its pseudogradient. There exists a $\delta>0$ such that for every $a^0\in\mathcal{A}\cap\pmb{\mathrm{B}}(a^*,\delta)$, where $\pmb{\mathrm{B}}(a^*,\delta)$ is the open ball with center at $a^*$ and radius $\delta$,  the Josephy-Newton update \eqref{eq:JN} generates a well-defined sequence $\{a^k\}$ in $\pmb{\mathrm{B}}(a^*,\delta)$, and every such sequence converges $Q$-quadratically to $a^*$. 
\end{thm}

\begin{remark}
    Compared to the Josephy-Newton method for nonlinear optimization \cite{cunis2024input}, here we do not require $F+\mathcal{N}_{\mathcal{A}}$ to be strongly regular at $a^*$. Instead, we only require $a^*$ to be stable (but not strongly stable). Note that since $a^*$ may not be strongly stable, the sequence $\{a^k\}$ satisfying \eqref{eq:JN} may not be unique, but every such sequence converges to $a^*$ $Q$-quadratically. In addition, we do not require any constraint qualification conditions here. In fact, the feasible set $\mathcal{A}_i$ may not necessarily be representable by a finite number of equality and inequality constraints. If $\mathcal{A}_i$ is indeed a finitely representable set, then the strict semicopositive condition in Theorems \ref{uniqueness}-\ref{Th:stability} may be replaced by a weaker condition while ensuring the stability of $a^*$ \cite[Section 5.3.1]{book_pang}. 
\end{remark}

\subsection{Perturbed Josephy-Newton method}\label{sec:per_JN}
Now we consider a perturbed version of the Josephy-Newton method: 
\begin{equation}\label{eq:JN_pertu}
    F(a^k,v^k)+H(a^k,v^k)(a-a^k)+\mathcal{N}_{\mathcal{A}}(a)\ni 0,
\end{equation}
where $F:\mathbb{R}^{n}\times\mathcal{V}\rightarrow\mathbb{R}^n$, $H:\mathbb{R}^{n}\times\mathcal{V}\rightarrow\mathbb{R}^{n\times n}$, and $\bold{v}=\{v^k\}_{k=0}^{\infty}\in\mathcal{V}$ is a disturbance sequence  that could model the inexact evaluation of the pseudogradient or a nonzero reminder in solving \eqref{eq:GE}. Assume that $\|\bold{v}\|_{\infty}=\sup_{k\in\mathbb{N}}\|v^k\|<\infty$, and $F(a,v)$ and $H(a,v)$ are both uniformly Lipschitz continuous in its arguments at $(a^*,0)$.

Eq. \eqref{eq:JN_pertu} can be treated as a perturbed dynamic system:
\begin{equation}\label{dynamic_system}
    a^{k+1}=\mathcal{T}(a^k,v^k).
\end{equation}

 Theorem \ref{tm:ISS_pert_NE} shows that this dynamic system is locally ISS. 

\begin{thm}\label{tm:ISS_pert_NE}
    Let $a^*$ be a stable NE of the game $\mathcal{G}(\mathcal{A},J)$ whose pseudogradient is $F$. There exist $\delta_a>0$, $\delta_v>0$, $L_a<1$ and $L_v<\infty$, such that for every $a^0\in\mathcal{A}\cap\pmb{\mathrm{B}}(a^*,\delta_a)$ and $\|\bold{v}\|_{\infty}<\delta_v$, the sequence $\{a^k\}$ generated by the perturbed Josephy-Newton update \eqref{eq:JN_pertu} satisfies
    \begin{equation}\label{eq:ISS_perturbed_Newton}
        \|a^{k+1}-a^*\|\leq L_a \|a^k-a^*\|+L_v\|v^k\|.
    \end{equation} 
    That is, \eqref{eq:JN_pertu} is locally input-to-state stable around $a^*$.
\end{thm}

\begin{cor}\label{C1}
   Let the conditions in Theorem \ref{tm:ISS_pert_NE} hold. The sequence $\{a^k\}$ generated by the perturbed Josephy-Newton update \eqref{eq:JN_pertu} satisfies not only \eqref{eq:ISS_perturbed_Newton} but also:
\begin{equation}\label{eq:ISS_perturbed_Newton_Q}
        \|a^{k+1}-a^*\|\leq L'_a \|a^k-a^*\|^2+L'_v\|v^k\|,
    \end{equation} 
    for some constants $L'_a<\infty$ and $L'_v<\infty$.
\end{cor}

\begin{remark}
    The strictly semicopositivity condition used to establish the stability of the NE (and thus Theorems \ref{tm:stability_NE} and \ref{tm:ISS_pert_NE} and Corollary \ref{C1}) is weaker than the strictly monotone game condition widely used in the literature \cite{Lacra_dissipativity,dominic_2}, which requires $F$ to be strictly monotone on $\mathcal{A}$. Note that $F$ is strictly monotone if and only if the game Hessian $H$ is positive definite on $\mathcal{A}$, leading to the satisfaction of the strict semicopositivity condition. However, a strictly semicopositive $H$ does not ensure a monotone $F$. For example, consider a two-player game on   $\mathbb{R}_{++}\times\mathbb{R}_{++}$ with costs $J_1=(1/2)x_1^2-3x_1x_2$ and $J_2=x_1x_2$. The pseudogradient $F(a)=\big[\begin{smallmatrix}
  1 & -3\\
  1 & 0
\end{smallmatrix}\big]a$ is \textbf{not} monotone on $\mathbb{R}_{++}\times\mathbb{R}_{++}$; while the game Hessian $H=\big[\begin{smallmatrix}
  1 & -3\\
  1 & 0
\end{smallmatrix}\big]$ is strictly semicopositive on $\mathbb{R}_{++}\times\mathbb{R}_{++}$. In other words, every strictly monotone game satisfies the strictly semicopositive condition, but the reverse is not true.
\end{remark}

\subsection{Distributed Josephy-Newton method}\label{sec:dis_JN}
Next we consider distributed algorithms where each agent performs a local computation  to update its own decision variable $a_i$ using its local information, i.e., $J_{i}$, $\nabla_{a_i} J_{i}$, and $\nabla^2_{a_ia_i} J_{i}$; but not $J_{-i}$ or its gradient or Hessian. Other agents' decision variables from previous rounds of iterations are available. Such distributed NE seeking is often designed based on better- or best- response dynamics \cite{Lacra_dissipativity,best_response}. We here show that the Josephy-Newton method can be used to solve the best response or proximal response problem, which, with appropriate monotonicity assumptions, leads to the NE \cite{best_response}.  

Given the current decision variables $(\Bar{a}_{i},\Bar{a}_{-i})$, agent $i$ seeks its proximal response to $\Bar{a}_{-i}$, denoted by $\Bar{a}_i^*$, as its next decision, i.e.,
\begin{equation}\label{eq:proximal_response}
\Bar{a}_i^*\in\argmin_{a_{i}\in\mathcal{A}_i}\left(J_i(a_i,\Bar{a}_{-i})+\frac{\tau}{2}\|a_i-\Bar{a}_i\|^2\right),
\end{equation}
where $\tau$ is a parameter of selection. Denote $\Bar{J}_i(a_i,\Bar{a}_{-i})=J_i(a_i,\Bar{a}_{-i})+\frac{\tau}{2}\|a_i-\Bar{a}_i\|^2$.

The optimization \eqref{eq:proximal_response} can be solved by the Newton update:
\begin{equation}\label{eq:distributed_JN}
    \nabla_{a_i}\Bar{J}_i(a_i^k,\Bar{a}_{-i})+\nabla^2_{a_ia_i}\Bar{J}_i(a_i^k,\Bar{a}_{-i})(a_i-a_i^k)+\mathcal{N}_{\mathcal{A}_i}(a_i)\ni 0.
\end{equation}

Theorem \ref{convergence_dis_local} shows that $a_i^k$ from \eqref{eq:distributed_JN} converges to $\Bar{a}_i^*$; and that the proximal response converges to a NE under a monotonicity assumption on $F$. 
\begin{thm}\label{convergence_dis_local}
    Let A1 and A2 hold. If  $\nabla^2_{a_ia_i}\Bar{J}_i(\Bar{a}_i^*,\Bar{a}_{-i})$ is strictly copositive on the critical cone $\mathcal{C}_i(\Bar{a}_i^*;\mathcal{A}_i,F^*_{i,{\Bar{a}_{-i}}})$ and $a_i^0$ is sufficiently close to $\Bar{a}_i^*$, then $\{a_i^k\}$ generated by  \eqref{eq:distributed_JN} converges to $\Bar{a}_i^*$ $Q$-quadratically. In addition, if $F$ is monotone on $\mathcal{A}$ and $\tau$ is large enough, then the proximal response $\{\Bar{a}_1^*,...,\Bar{a}_i^*\}$ by iterating $i$ converges to a NE of the game $\mathcal{G}(\mathcal{A},J)$. 
\end{thm}

The ISS property of the distributed Josephy-Newton \eqref{eq:distributed_JN} follows from a similar analysis as in Section \ref{sec:per_JN}.

\section{Newton Methods for GNE Problems}\label{sec:GNE_section}
This section considers the GNE problem \eqref{GNE_problem}.
In this section, we assume that $\mathcal{A}_i(a_{-i})$ can be represented by a finite number of inequality constraints, i.e., $\mathcal{A}_i(a_{-i})=\{a_i\in\mathbb{R}^{n_{i}}: g_i(a_i,a_{-i})\leq0\}$, where $g_i:\mathbb{R}^{n}\rightarrow \mathbb{R}^{m_i}$. The game \eqref{GNE_problem} can then be equivalently written as 
\begin{subequations}\label{GNE_2}
    \begin{align}
        \min_{a_{i}}\quad & J_i(a_i,a_{-i})\\
         s.t.  \quad &
          g_i(a_i,a_{-i})\leq0.
    \end{align}
\end{subequations}

If $a^*=\{a_1^*,a_2^*,..,a_N^*\}$ solves \eqref{GNE_2} for all $i\in\mathcal{N}$, then $a^*$ is a GNE \cite{GNEP}. The Lagrangian associated with \eqref{GNE_2} is
\begin{equation}
L_i(a_i,a_{-i},\lambda_i)=J_i(a_i,a_{-i})+\lambda_i^Tg_i(a_i,a_{-i}),
\end{equation}
where $\lambda_i\in\mathbb{R}^{m_i}_+$ are dual variables. Denote $m=\sum_{i=1}^N m_i$.

If $a_i^*$ is an optimal solution to \eqref{GNE_2} given $a_{-i}^*$, and if a suitable constraint qualification holds, there exists $\lambda_i^*\in\mathbb{R}^{m_i}_+$ such that $(a_i^*,\lambda_i^*)$ satisfy the following KKT conditions \cite{GNEP_KKT}:
\begin{equation}\label{KKT_i}
    \begin{pmatrix}
\nabla_{a_i}L_i(a_i,a_{-i}^*,\lambda_i)\\
-g_i(a_i,a_{-i}^*)
\end{pmatrix}
+\mathcal{N}_{\mathbb{R}^{n_{i}}\times\mathbb{R}^{m_i}_+}(a_i,\lambda_i)\ni0
\end{equation}

 Concatenating these $N$ KKT systems gives
\begin{equation}\label{KKT_concatenated}
    \begin{pmatrix}
\mathbf{L}(a^*,\lambda^*)\\
-\mathbf{g}(a^*)
\end{pmatrix}
+\mathcal{N}_{\mathbb{R}^{n}\times\mathbb{R}^{m}_+}(a^*,\lambda^*)\ni0
\end{equation}
where 
 \begin{equation}\nonumber
 \mathbf{L}(a,\lambda)=\begin{pmatrix}
           \nabla_{a_1}L_1(a,\lambda_1) \\
           \vdots \\
           \nabla_{a_N}L_N(a,\lambda_N)
         \end{pmatrix},
 \end{equation}
\begin{equation}\nonumber
   \mathbf{g}(a)=\begin{pmatrix}
           g_{1}(a) \\
           \vdots \\
           g_{N}(a)
         \end{pmatrix}, \quad \text{and} \quad
            \lambda=\begin{pmatrix}
           \lambda_1 \\
           \vdots \\
           \lambda_N
         \end{pmatrix}. 
\end{equation}

We make the following assumption throughout this section. 

A3. (Smoothness): For every $i$, the functions $J_i$ and $g_i$ in \eqref{GNE_2} are twice continuously differentiable in their arguments, and their second-order derivatives are Lipschitz continuous.

\subsection{A special class of GNEs: Variational-GNE}\label{sec:v-GNE}
Due to the complexity of the GNE problems, most existing studies focus on a special class of GNEs: v-GNE \cite{dominic_2,Lacra_distributed2}, which often requires the following assumptions. 

A4. (Player-wise convexity): For every $i$  and $a_{-i}$, $J_i(\cdot,a_{-i})$ is convex, and $\mathcal{A}_i(a_{-i})$ is compact and convex;

A5. (Common constraints): Agents share common constraints: $g_1=...=g_N=\overline{g}$. 

Under Assumption A5, we have $m_1=...=m_N=\overline{m}$ and that \eqref{KKT_concatenated} becomes
\begin{equation}\label{KKT_concatenated_common}
    \begin{pmatrix}
\mathbf{L}(a^*,\lambda^*)\\
\begin{Bmatrix}-\overline{g}(a^*)\\
\vdots\\
-\overline{g}(a^*)
\end{Bmatrix}\text{$N$ times}
\end{pmatrix}
+\mathcal{N}_{\mathbb{R}^{n}\times\mathbb{R}^{N\overline{m}}_+}(a^*,\lambda^*)\ni0.
\end{equation}

A v-GNE \cite{GNEP} is defined as the primal part of $(a^*,\overline{\lambda}^*)$ such that 
\begin{equation}\label{KKT_concatenated_common_lambda}
    \begin{pmatrix}
\mathbf{L}(a^*,\lambda^*)\\
-\overline{g}(a^*)\\
\end{pmatrix}
+\mathcal{N}_{\mathbb{R}^{n}\times\mathbb{R}^{\overline{m}}_+}(a^*,\overline{\lambda}^*)\ni0.
\end{equation}
Note that $\lambda\in\mathbb{R}^{N\overline{m}}_+$, while $\overline{\lambda}\in\mathbb{R}^{\overline{m}}_+$. 

If $(a^*,\overline{\lambda}^*)$ solves \eqref{KKT_concatenated_common_lambda}, then $(a^*,\lambda^*)$, where $\lambda_i^*=\overline{\lambda}^*$ for all $i$,  must solve \eqref{KKT_concatenated_common}. Further, by noting that \eqref{KKT_concatenated_common_lambda} is the KKT condition of VI$(\mathcal{A},F)$ with $F$ defined by \eqref{P_gradient} and $\mathcal{A}=\{a\in\mathbb{R}^{n}: \overline{g}(a_i,a_{-i})\leq0\}$, we can say that if $a^*$ solves VI$(\mathcal{A},F)$ and if A3-A5 hold, then $a^*$ must be a GNE of \eqref{GNE_2}. 
Due to this equivalence, one may employ the Newton-type methods in Section \ref{sec:NE} directly for the v-GNE seeking.

\subsection{Semismooth Newton methods}\label{sec:GNE}
In this subsection we consider solving \eqref{KKT_concatenated} without A4 and A5. Note that \eqref{KKT_concatenated} is not equivalent to a KKT system of \textbf{\textit{any}} VIs. Instead, it corresponds to a quasi-VI problem, QVI$(\mathcal{A}(a),F)$, in which the feasible set is not fixed but depends on the point being evaluated. We start with the coupled KKT systems \eqref{KKT_concatenated}, use complementarity functions to transform them to nonsmooth equations, and  employ semismooth Newton methods to solve these nonsmooth equations.

A complementarity function $\phi:\mathbb{R}^2\rightarrow\mathbb{R}$ is a function such that $
    \phi(x,y)=0 \Longleftrightarrow  x\geq0, y\geq0, xy=0.$
If $\phi$ is a complementarity function, then by denoting 
\begin{equation}
    \pmb{\phi}(-g(a),\lambda)=\begin{pmatrix}
           \phi(-g_{1}(a),\lambda_1) \\
           \vdots \\
           \phi(-g_{N}(a),\lambda_N)
         \end{pmatrix},
\end{equation}
the KKT system \eqref{KKT_concatenated} can be reformulated as 
\begin{equation}\label{system_equations}
     \Phi(z)=\begin{pmatrix}
          \mathbf{L}(z)\\
          \pmb{\phi}(-g(a),\lambda) \\
           \end{pmatrix}=0,
\end{equation}
where $z=(a,\lambda)$.

One widely-used complementarity function is the $min$ function:
\begin{equation}
    \phi(x,y)=\min\{x,y\} \quad \forall x,y\in\mathbb{R}.
\end{equation}
Because this $min$ function is a strongly semismooth function and because the composition of strongly semismooth functions is strongly semismooth \cite[Proposition 7.4.4]{book_pang}, the function $\Phi$ in \eqref{system_equations} is strongly semismooth everywhere. 

We now consider the semismooth Netwon method to solve the system of equations \eqref{system_equations}. At the iteration $k$, one solves the following linear equation to find $z^{k+1}$:
\begin{equation}\label{semi_newton}
    \Phi(z^{k})+J\Phi(z^k)(z-z^k)=0,
\end{equation}
where $J\Phi(z^k)\in \text{Jac }\Phi(z^k)$, with $\text{Jac }\Phi(z^k)$ being the limiting Jacobian of $\Phi$ at $z^k$. 
 
To ensure that the sequence generated by \eqref{semi_newton} is well-defined, the following quasi-regularity condition is needed.

\begin{defn}\cite[Definition 3]{GNEP_newton}\label{defn:quasi-regular}
    A point $z^*=(a^*,\lambda^*)$ is \textit{quasi-regular} if all matrices in  $\text{Jac }\Phi(z^*)$ are nonsignular.
\end{defn}

\begin{remark}
    According to \cite{quasirefular_vs_stronglyregular}, this quasi-regular condition is weaker than the strongly regular condition in nonlinear optimization problems.  
\end{remark} 

\begin{thm}\cite[Theorem 2]{GNEP_newton}\label{thm:semismooth-conv}
    Let assumption A3 hold. Let $z^*=(a^*,\lambda^*)$ be a quasi-regular solution of system \eqref{system_equations}. Then a neighborhood $\Omega$ of $z^*$ exists so that if $z^0\in\Omega$ then the semismooth Newton update \eqref{semi_newton} is well defined and generates a sequence $\{z^k\}$ that converges $Q$-quadratically to $z^*$. 
\end{thm}

\subsection{Perturbed semismooth Newton methods}\label{subsection:PSSN}
Next we consider a perturbed version of \eqref{semi_newton} as
\begin{equation}\label{eq:inexact_seminewton}
    \Phi(z^{k})+J\Phi(z^k)(z-z^k)=r^k,
\end{equation}
where $r^k$ is a sequence of disturbances with $\|r^k\|<\infty$ for all $k$. The next theorem shows that \eqref{eq:inexact_seminewton} is locally ISS around $z^*$ with $Q$-quadratic convergence.
\begin{thm}\label{thm:ISS_semi}
    Let assumption A3 hold. Let $z^*$ be a quasi-regular solution to \eqref{system_equations}. There exist $\delta_z>0$, $\delta_r>0$, $L_z<\infty$ and $L_r<\infty$ such that for every $z^0\in\pmb{\mathrm{B}}(z^*,\delta_z)$ and $\|\bold{r}\|_{\infty}=\sup_{k\in\mathbb{N}}\|r^k\|<\delta_r$, the sequence $\{z^k\}$ generated by the perturbed semismooth Newton method \eqref{eq:inexact_seminewton} satisfies
    \begin{equation}
        \|z^{k+1}-z^*\|\leq L_z\|z^k-z^*\|^2+L_r\|r^k\|. 
    \end{equation}
\end{thm}

\subsection{Distributed semismooth Newton methods}\label{sec:distributed_SS_Newton_GNE}
Similar to the NE problem, the distributed GNE seeking algorithm can also be designed based on proximal response dynamics. That is, given the current decision variables $(\Bar{a}_{i},\Bar{a}_{-i})$, agent $i$ solves the following constrained optimization problem to find its next decision: 
\begin{subequations}\label{eq:Semi_proximal}
    \begin{align}
        \min_{a_{i}}\quad & J_i(a_i,\Bar{a}_{-i})+\frac{\tau}{2}\|a_i-\Bar{a}_i\|^2\\
         s.t.  \quad &
          g_i(a_i,\Bar{a}_{-i})\leq0.
    \end{align}
\end{subequations}

The KKT conditions of \eqref{eq:Semi_proximal} can be solved by the semismooth Newton update with guaranteed convergence and ISS following a similar analysis as in Sections \ref{sec:GNE} and \ref{subsection:PSSN}. However, we comment that the convergence of the proximal response in a GNE problem requires more restrictive conditions on the game structure. One sufficient condition, for example, is that the game is a generalized potential game \cite{GPG_def}.

\section{Application to CG-MPC}\label{sec:CG-MPC}
We apply the ISS results to CG-MPC problems to enable time-distributed solution seeking and real-time computation. 
\subsection{CG-MPC problem}\label{sec:CG_MPC}
 Consider a multi-agent system with dynamics
\begin{equation}\label{equ:basic_dynamics}
    x(t+1)=f(x(t),u(t)),
\end{equation}
where $x=\{x_1,...,x_N\}\in\mathcal{X}$, $u=\{u_1,...,u_N\}\in\mathcal{U}$, $x_i\in\mathcal{X}_i\subset\mathbb{R}^{n_{x_i}}$ and $u_i\in\mathcal{U}_i\subset\mathbb{R}^{n_{u_i}}$. 

Each agent is controlled using MPC to optimize its self-interest subject to coupled constraints. Mathematically, a CG-MPC problem at each  $t$ is formulated as:
\begin{subequations}\label{MPC_1}
    \begin{align}
        \min_{a_i} \quad &J_i(a_i,a_{-i})=\Psi_i(\xi(T))+\sum_{\tau=0}^{T}l_i(\xi(\tau),a(\tau)),\\
         s.t., \quad &\xi(0)=x(t),\\  \quad  &\xi(\tau+1)=f(\xi(\tau),a(\tau)), \\
        &  (\xi_i(\tau),a_i(\tau))\in\mathcal{Z}_i(\xi_{-i}(\tau),a_{-i}(\tau))\label{62d},
    \end{align}
\end{subequations}
where $T\in\mathbb{N}_{++}$, $\tau=0,1,...T$, $\xi_i=\{\xi_i(0),\xi_i(1),...,\xi_i(T)\}$, $\xi=\{\xi_1,...,\xi_N\}$,  $a_i=\{a_i(0),a_i(1),...,a_i(T)\}$, and $a=\{a_1,...,a_N\}$. 
For simplicity, all functions in \eqref{MPC_1} are assumed to be twice continuously differentiable in their arguments, and their second order derivatives are Lipschitz continuous. 

The CG-MPC problem \eqref{MPC_1} can be written compactly as the following GNE problem parameterized by $x(t)$:
\begin{subequations}\label{game_2}
    \begin{align}
        \min_{a_{i}}\quad & J_i^{x(t)}(a_i,a_{-i})\\
         s.t.\quad & a_{i}\in\mathcal{A}^{x(t)}_i(a_{-i}).\label{7b}
    \end{align}
\end{subequations}

In cases where agent $i$'s feasible set does not depend on $a_{-i}$ (i.e., \eqref{7b} becomes $a_{i}\in\mathcal{A}^{x(t)}_i$), the game \eqref{game_2} reduces to a NE problem.

\subsection{Time-distributed CG-MPC solver}\label{sec:TDO}
Instead of solving \eqref{game_2} completely at each sampling instant, we consider approximately tracking the solution trajectory of \eqref{game_2} as $x$ evolves. To this end, we distribute the iterations of the CG-MPC solver over time and warmstart the approximate solution at $t+1$ using the result at $t$.  Specifically, the approximate solution at $t+1$ is generated by
\begin{align}
        v(t+1)&=\mathcal{T}^{K,x(t+1)}(v(t)),\label{eq:TDO}
\end{align}
where $v$ is an estimate of $a$ or $z$ (decision variable or primal-dual variable, depending on the algorithm design), and $\mathcal{T}^{K,x(t)}$ represents a fixed number (i.e., $K$) of solver iterations parametrized by $x(t)$. The iteration \eqref{eq:TDO} is itself a dynamic system evolving in parallel to the controlled plant. The iteration algorithm $\mathcal{T}^{K,x(t)}$ should be designed such that when $K\rightarrow \infty$, $v(t)=v^*(t)$, where $v^*$ denotes the exact solution of the CG-MPC problem at $t$. The Newton-type methods developed in Sections \ref{sec:NE} and \ref{sec:GNE_section} can serve as such algorithms. That is, $\mathcal{T}^{K,x(t)}$ means $K$ times of consecutive updates \eqref{eq:JN} in NE problems with $F$, $H$, and $\mathcal{A}$ being parameterized by $x(t)$; 
or $K$ times of consecutive updates \eqref{semi_newton} in GNE problems 
with $\Phi$ and $J\Phi$ being parameterized by $x(t)$. Such iterations may also be performed in an agent-distributed manner: 
\begin{equation}\label{eq:distributed_TDO}
        v_i(t+1)=\mathcal{T}_i^{K,x(t+1)}(v_i(t)), 
\end{equation}
with $v=\{v_1,\cdots,v_N\}$ and $\mathcal{T}=\{\mathcal{T}_1,\cdots,\mathcal{T}_N\}$, where $\mathcal{T}_i$ can be designed based on the proximal response dynamics.

Next we show that the approximate solution from \eqref{eq:TDO} tracks the optimal solution with bounded tracking error: 
\begin{equation}\label{eq:error}
    e(t)=v(t)-v^*(t).
\end{equation}

\begin{thm}\label{thm:bounded_error}
    Let $x$ and $v^*$ be Lipschitz continuous with respect to $t$ and let $v^*$ be isolated. Let $\mathcal{T}$ be designed according to  \eqref{eq:JN}  or \eqref{semi_newton}  with appropriate conditions as specified in Theorems \ref{tm:ISS_pert_NE} and \ref{thm:ISS_semi}, respectively. Then there exist $\alpha,\theta:\mathbb{N}_{++}\rightarrow\mathbb{R}_{++}$ both of class $\mathcal{L}$ and $e(0)$ sufficiently small so that the $e(t)$ satisfies 
    \begin{equation}\nonumber
        \|e(t+1)\|\leq \alpha(K)\|e(t)\|+\theta(K)\|\Delta x(t)\|,
    \end{equation}
    where $\Delta x(t)=x(t+1)-x(t)$.
\end{thm}

\begin{remark}
    Theorem \ref{thm:bounded_error} states that the tracking error $e(t)$ of the CG-MPC solver remains bounded for all $t$ and could be made arbitrarily small by increasing $K$. Other closed-loop properties, such as the ISS of the interconnected plant-solver system against external disturbances to the plant, and the constraint satisfaction, can be established analogously to \cite{liao2020time}, thanks to the $Q$-quadratic convergence and the ISS property of the developed Newton-type methods.
\end{remark}

\subsection{Numerical studies}
Consider a dynamic system of $3$ agents: 
\begin{equation}\label{eq:dynamic}
\begin{split}
    x(t+1)&=f(x(t),u(t))\\
    &=\begin{bmatrix}
0.5 & 0 & 0\\
0 & 0.3 & 0\\
0 & 0 & 0.2
\end{bmatrix}x(t)+\begin{bmatrix}
1 & 0 & 0\\
0 & 1 & 0\\
0 & 0 & 1
\end{bmatrix}u(t),
\end{split}
\end{equation}
where $x=[x_1,x_2,x_3]^T$ and $u=[u_1,u_2,u_3]^T$. At each $t$, the following GNE problem needs to be solved with $a_{i}=\{a_{i}(0),...,a_{i}(T)\}$, $a=[a_1,a_2,a_3]^T$, and $\xi=[\xi_1,\xi_2,\xi_3]^T$:
\begin{equation}\label{eq:player1}
    \begin{split}
\text{Agent 1:}    \quad    \min_{a_{1}}\quad & \sum_{\tau=0}^{T} (\xi_1(\tau)-1)^2-\xi_1(\tau)\xi_2(\tau) \\
         s.t.  \quad & \xi(0)=x(t),\\
         &\xi(\tau+1)=f(\xi(\tau),a(\tau)), \\
         & a_1(\tau)+a_2(\tau)+a_3(\tau)\leq 1, \quad \forall\tau.\\
\text{Agent 2:}    \quad        \min_{a_{2}}\quad & \sum_{\tau=0}^{T} (\xi_2(\tau)-1)^2+\xi_1(\tau)\xi_2(\tau) \\
         s.t.  \quad & \xi(0)=x(t),\\
         &\xi(\tau+1)=f(\xi(\tau),a(\tau)), \\
         & a_2(\tau)\leq \frac{1}{2}, \quad \forall\tau.\\
\text{Agent 3:}    \quad        \min_{a_{3}}\quad & \sum_{\tau=0}^{T} (\xi_3(\tau)-1)^2 \\
         s.t.  \quad & \xi(0)=x(t)\\
         &\xi(\tau+1)=f(\xi(\tau),a(\tau)), \\
         &a_3(\tau)\geq a_1(\tau)+a_2(\tau), \quad \forall\tau\\
         \quad & a_3(\tau)\geq0, \quad \forall\tau.
    \end{split}
\end{equation}

\textbf{Study 1:} We first solve the GNE problem \eqref{eq:player1} at one specific $t$ (e.g., $t=0$). Because of the coupled constraints, we use the semismooth Newton method developed in Section \ref{sec:GNE}. The convergence of $z^k=\{a^k,\lambda^k\}$ to $z^*=\{a^*,\lambda^*\}$ with the update rule \eqref{semi_newton} and the disturbed update \eqref{eq:inexact_seminewton} are shown in Fig. \ref{errornorm} and Fig. \ref{errornorm_disturbed} respectively. For the disturbed update, a random disturbance is added to each element of $J\Phi(z^k)$, $\forall k$. As shown in Fig. \ref{Newton_fix}, the semismooth Newton method solves the KKT system of \eqref{eq:player1} with fast convergence (Fig. \ref{errornorm}), and is robust to external disturbances (Fig. \ref{errornorm_disturbed}).

\begin{figure}[thpb]
\centering
\subfigure[]{\label{errornorm}
\includegraphics[width=0.23\textwidth]{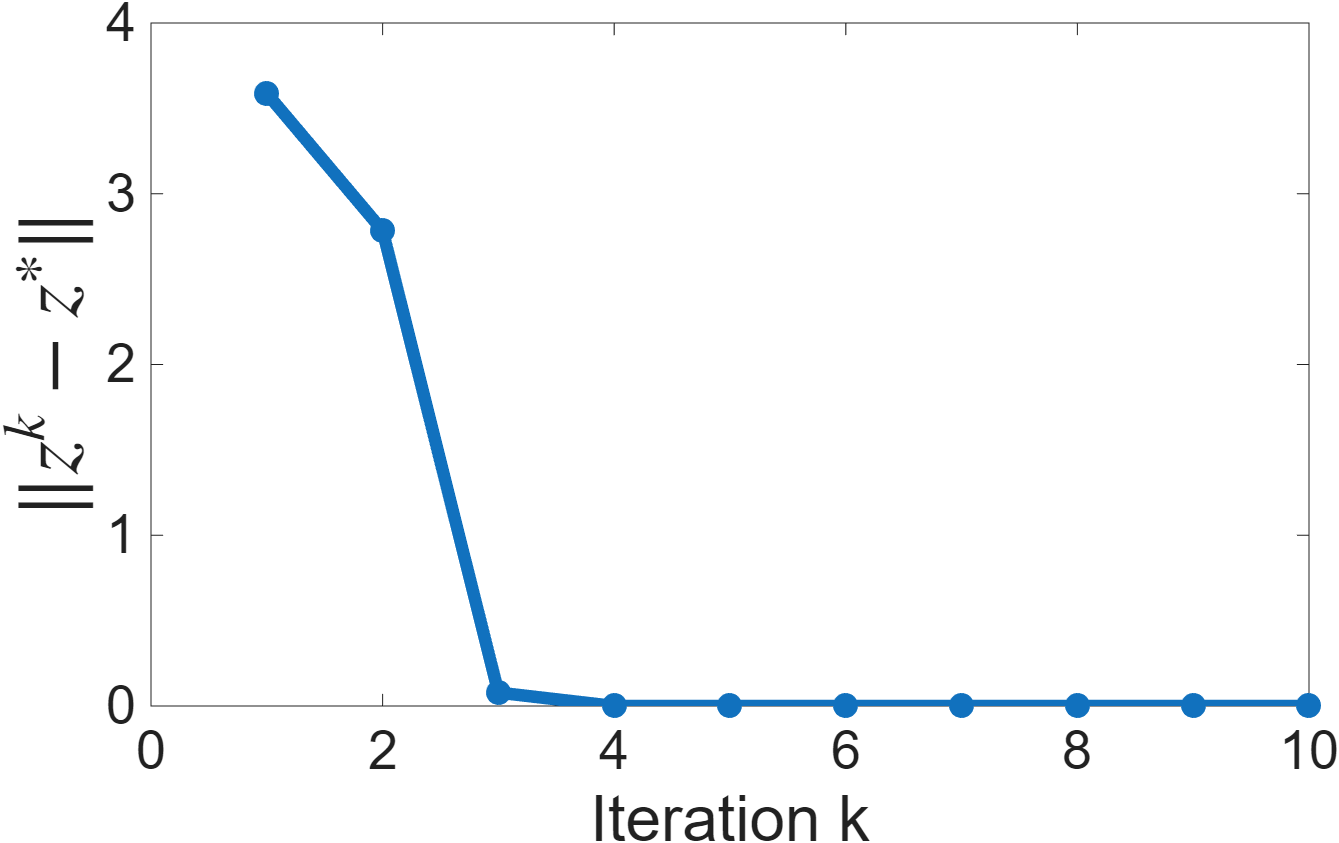}}
\subfigure[]{\label{errornorm_disturbed}
\includegraphics[width=0.23\textwidth]{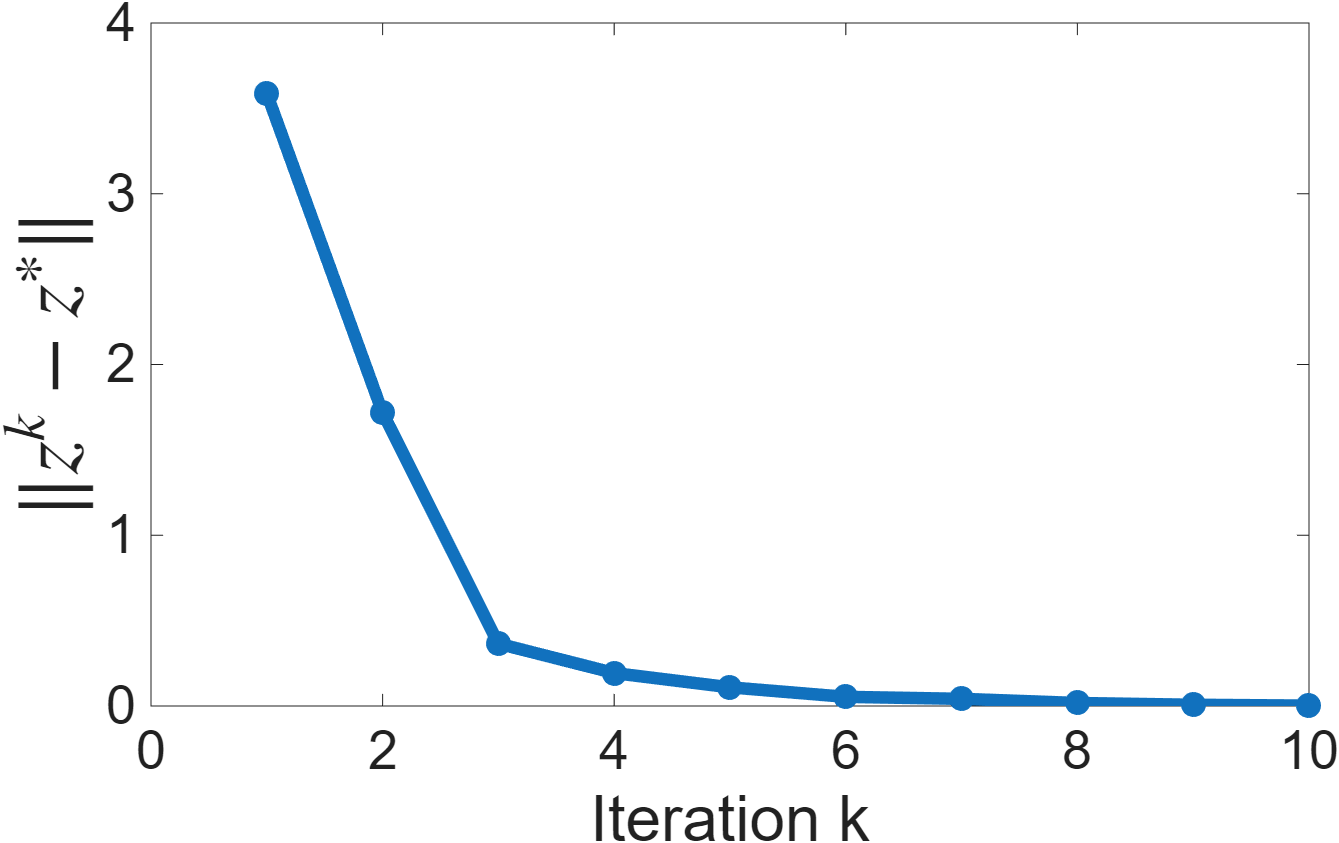}}
\caption{$\|z^k-z^*\|$ with $z^k$ generated from (a) semismooth Newon method \eqref{semi_newton} and (b) disturbed semismooth Newon method \eqref{eq:inexact_seminewton}, respectively.}\label{Newton_fix}
\end{figure}

\textbf{Study 2:} We then consider the CG-MPC problem \eqref{eq:player1} for consecutive $t$ and with time-distributed iterations. Specifically, given $x(0)=[1,3,5]^T$, the CG-MPC solver generates a sequence $\{z^1(0),...,z^K(0)\}$ according to \eqref{semi_newton}, where each element $z^k(0)=\{z^{k,\tau}(0)\}_{\tau\in[0,T]}$ represents a time sequence of the primal-dual variables with $k\in[1,K]$ and $K$ being the maximum iterations at each $t$. The primal part of $z^{K,0}(0)$ is then applied to the plant \eqref{eq:dynamic}, which leads to $x(1)$; the solver then seeks $\{z^1(1),...,z^K(1)\}$  following the same procedure. Note that a sufficiently large $K$ means that the KKT system of \eqref{eq:player1} is solved completely at each $t$, while a small $K$ means possibly non-zero error $e(t)=z^K(t)-z^*(t)$. Figures \ref{state} and \ref{error_dynamic} show $x(t)$ and $\|e(t)\|=\|z^K(t)-z^*(t)\|$ for different $K$, respectively. It is observed that when $K=7$, $\|e(t)\|$ remains $0$ for all $t$, suggesting that $7$ iterations of Newton updates are sufficient to solve the KKT system of \eqref{eq:player1} at each $t$.  As $K$ decreases, the errors at the first few time instants increase, indicating incomplete optimization. As $t$ increases, the state converges to constants, and the error goes to $0$ regardless of $K$, as suggested by Theorem \ref{thm:bounded_error}.   

\begin{figure}[thpb]
\centering
\subfigure[]{\label{state}
\includegraphics[width=0.23\textwidth]{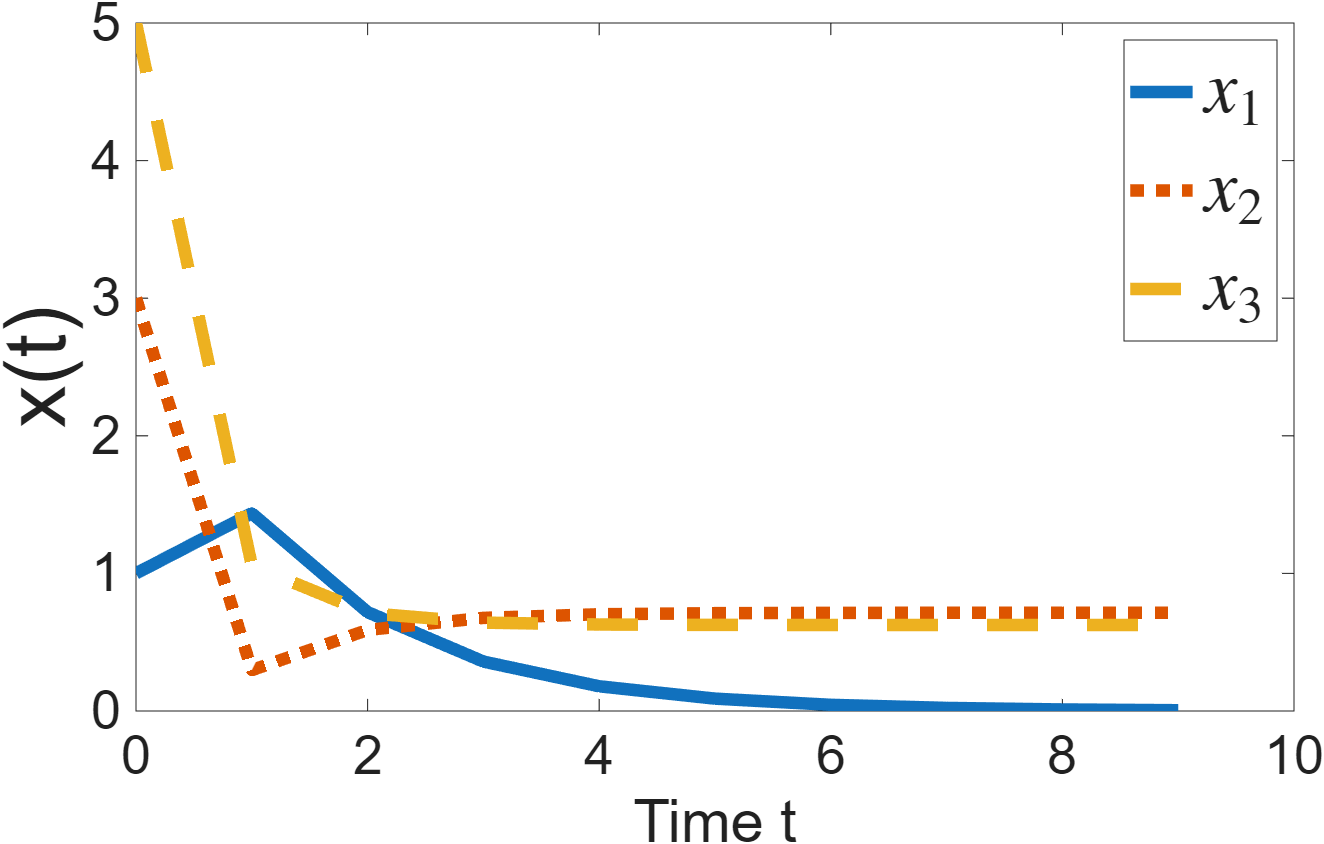}}
\subfigure[]{\label{error_dynamic}
\includegraphics[width=0.23\textwidth]{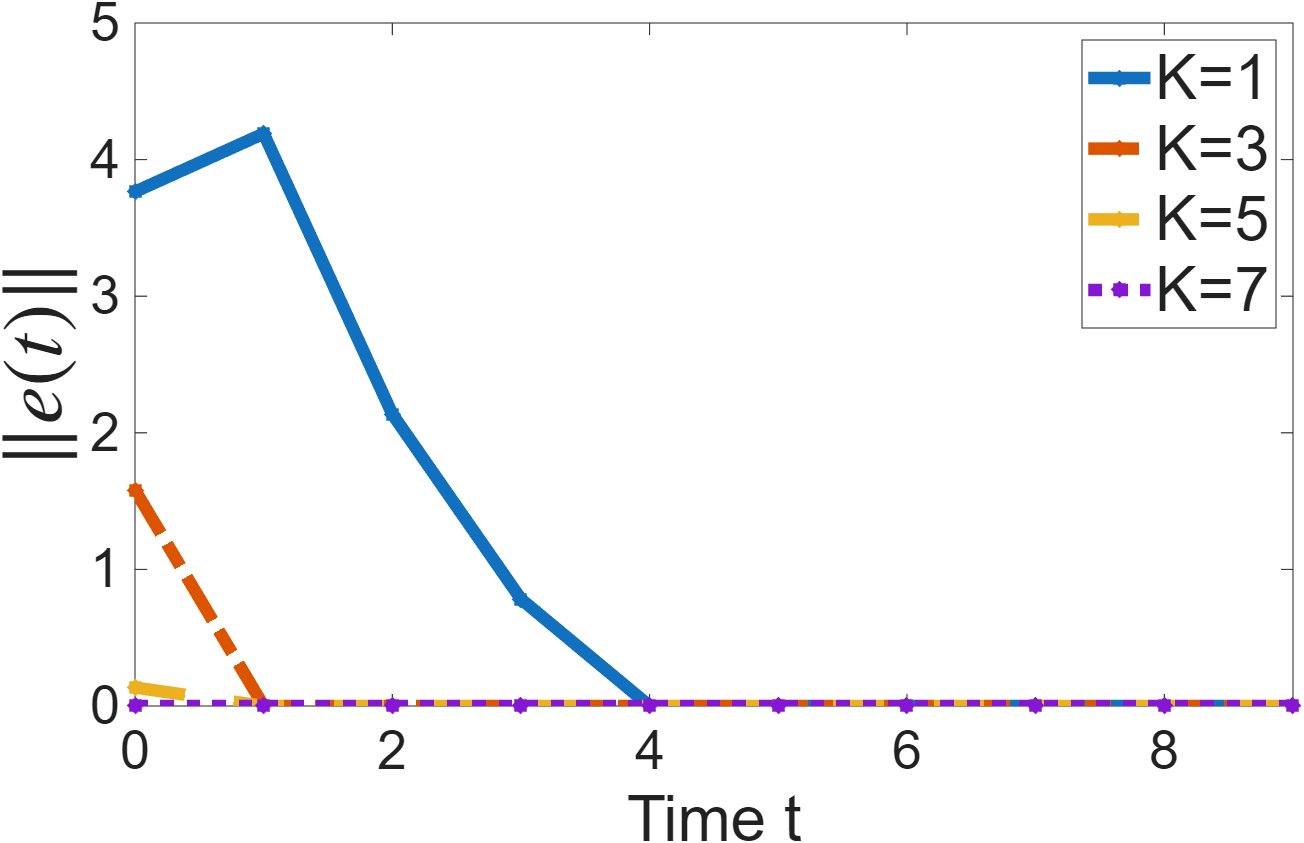}}
\caption{CG-MPC with time-distributed iterations. (a) $x(t)$ generated by $z^*(t)$, and (b) Error norm $\|e(t)\|=\|z^K(t)-z^*(t)\|$ with different $K$.}\label{Newton_dynamic}
\end{figure}

\textbf{Study 3:} Next we consider the agent-distributed implementation. Specifically, at each $t$, agent $i$ solves its proximal response according to \eqref{eq:Semi_proximal} until convergence, and the KKT condition of \eqref{eq:Semi_proximal} is solved by the semismooth Newton method. Other procedures are the same as in Study 2. Figs. \ref{distributed2} and \ref{distributed5} show $\|e_i(t)\|=\|z_i^K(t)-z_i^*(t)\|$ when $K=2$ and $K=5$, respectively. It is observed that for both values of $K$ and for all $i$, $e_i(t)$ goes to $0$ as $t$ increases, indicating the convergence of the proximal response algorithm; and that $K=5$ leads to a much faster convergence of $e_i(t)$ compared to $K=2$, consistent with the observations in Study 2.

\begin{figure}[thpb]
\centering
\subfigure[]{\label{distributed2}
\includegraphics[width=0.23\textwidth]{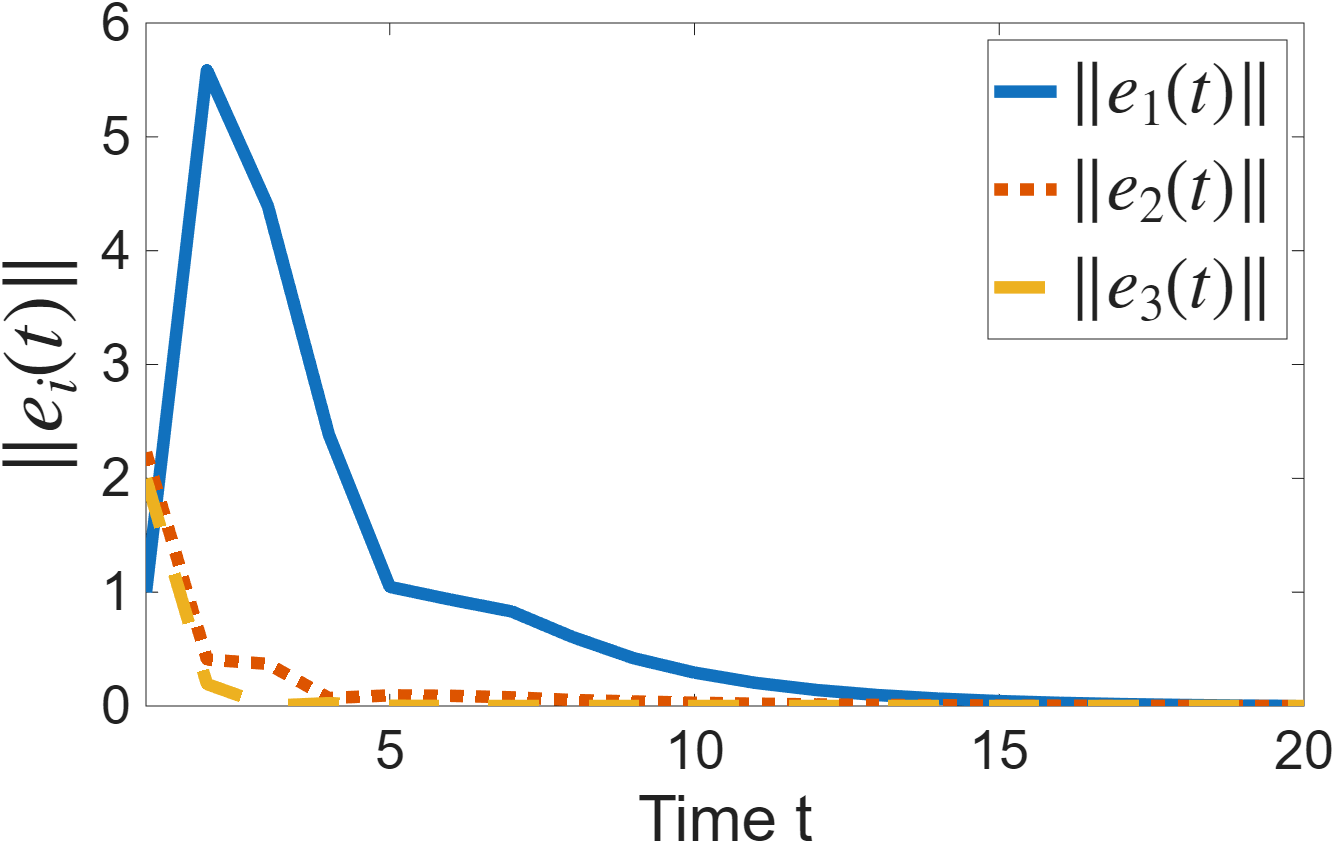}}
\subfigure[]{\label{distributed5}
\includegraphics[width=0.23\textwidth]{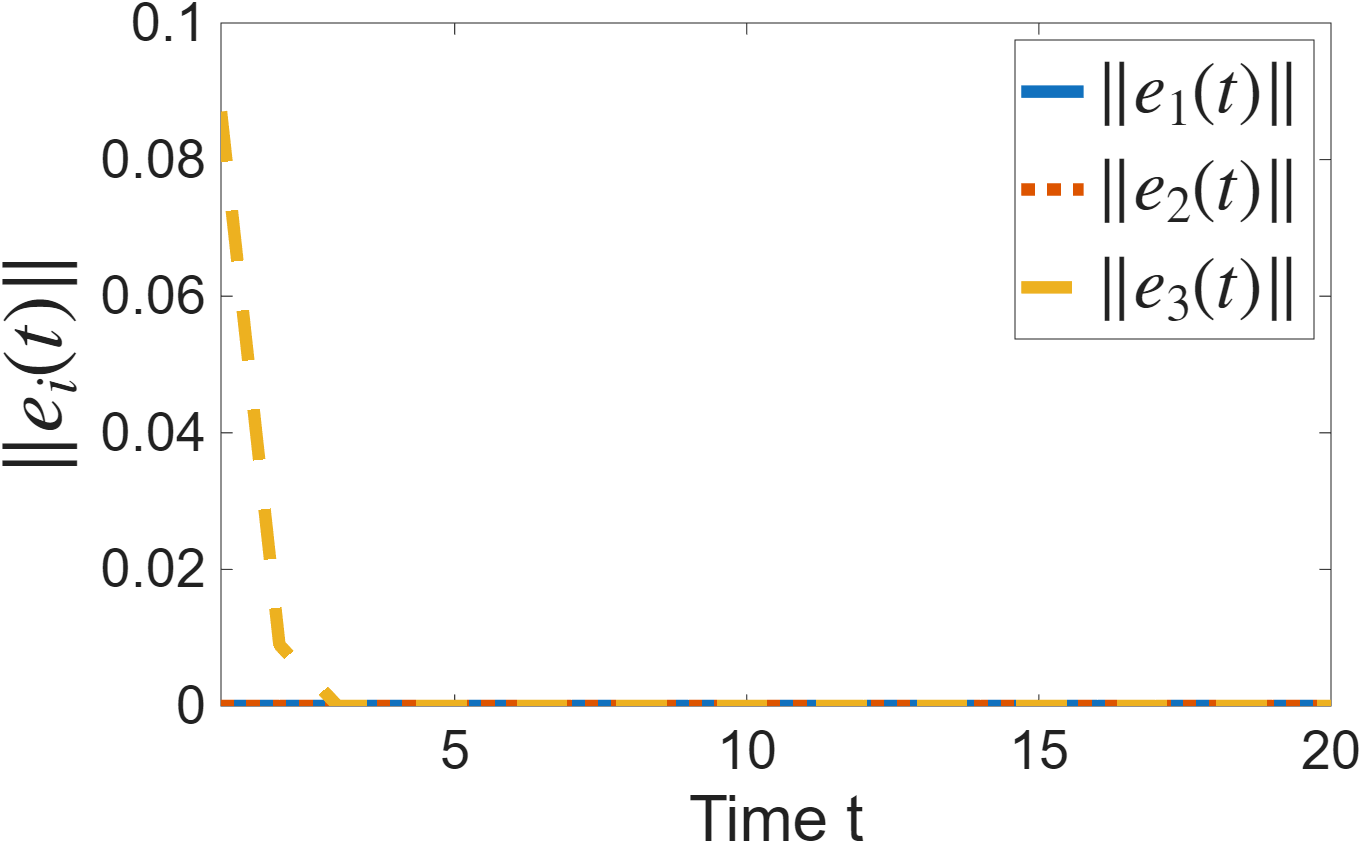}}
\caption{Error norms $e_i(t)=z_i^K(t)-z_i^*(t)$ for each agent using agent-distributed CG-MPC when (a) $K=2$ and (b) $K=5$ respectively.}
\end{figure}

\section{Conclusion}\label{sec:conclusion}
Newton-type methods have been widely used in nonlinear optimization, and this paper shows that their effectiveness is retained in NE and GNE problems. In NE problems, the convergence of Josephy-Newton method and the ISS of the perturbed iterations were proven under the NE stability condition, weaker than the strong regularity condition typically used in nonlinear optimization, and weaker than the strictly monotone game condition typically used in the NE-seeking literature. For GNE problems, semismooth Newton methods were developed to solve the KKT systems. The ISS was proven under a quasi-regularity condition, also weaker than the strong regularity condition. 
A CG-MPC solver with time-distributed iterations was developed for real-time implementation. Boundness of tracking error was proven and was verified by numerical analysis.

\section*{Appendix}
\subsection{Proof of Theorem \ref{Th:stability}}
For every $\Tilde{\mathcal{G}}\in\pmb{\mathrm{B}}(\mathcal{G};\epsilon,\mathcal{A}\cap \mathsf{cl}\Omega)$, its local solvability (i.e., \eqref{eq:local_solvability}) follows from  \cite[Proposition 12.16]{pang_bookchapter}. 
The bounded solution deviation (i.e., \eqref{eq:error_def}) follows from \cite[Corollary 5.1.8]{book_pang} by noting that $F$ must be at least B-differentiable \cite[Definition 3.1.2]{book_pang} at $a^*$ under assumption A1.

\subsection{Proof of Theorem \ref{tm:stability_NE}}
Define the error function $\mathbf{e}_F(a)$ as the residual of the first-order Taylor expansion of $F$ at $a$ relative to $a^*$:
\begin{equation}\label{eq:residual_error}
    \mathbf{e}_F(a)=F(a)-F(a^*)-H(a^*)(a-a^*).
\end{equation}
If $H$ is Lipschitz continuous in a neighborhood of $a^*$, then according to \cite[Proposition 7.2.9]{book_pang},
\begin{equation}\label{eq:q_qua}
    \lim_{a^*\neq a\rightarrow a^*}\sup\frac{\|\mathbf{e}_F(a)\|}{\|a-a^*\|^2}<\infty.
\end{equation}

Denote $F^k(a)=F(a^k)+H(a^k)(a-a^k)$ and $F^*(a)=F(a^*)+H(a^*)(a-a^*)$.
By the continuous differentiability of $F$ near $a^*$, the affine function $F^*$ is a strong first-order approximation (FOA) of $F$ at $a^*$ \cite[Section 5.2.2]{book_pang}. Because $a^*$ is a stable NE of the game $\mathcal{G}(\mathcal{A},J)$, it is a stable solution to VI$(\mathcal{A},F)$. Because $F^*$ is a strong FOA of $F$ at $a^*$, $a^*$ is also a stable solution to VI$(\mathcal{A},F^*)$ \cite[Proposition 5.2.15]{book_pang}. 
 The remaining proof is then analogous to \cite[Theorem 7.3.5]{book_pang} by noting that \eqref{eq:JN} is equivalent to solving VI$(\mathcal{A},F^k)$, and that \eqref{eq:q_qua} results in
$
\lim_{\substack{k\rightarrow\infty\\a^k\neq a^*}}\sup\frac{\|a^{k+1}-a^{*}\|}{\|a^{k}-a^{*}\|^2}<\infty$, leading to $Q$-quadratic convergence.

\subsection{Proof of Theorem \ref{tm:ISS_pert_NE}}
According to \eqref{eq:JN_pertu}, $a^{k+1}$ is a solution to VI$(\mathcal{A},F^{v,k})$, where
\begin{equation}\label{eq:F^k}
    F^{v,k}(a)=F(a^k,v^k)+H(a^k,v^k)(a-a^k).
\end{equation}
Denote
\begin{align}
    &F^*(a)=F(a^*,0)+H(a^*,0)(a-a^*),\label{eq:F^*}\\
    &\mathbf{e}_{F}(a,v)=F(a,v)-F(a^*,0)-H(a^*,0)(a-a^*). \label{eq:residual_error_disturb}
\end{align}
Due to the stability of $a^*$, we have
    \begin{equation}\label{eq:akv-a*}
    \begin{split}
        \|a^{k+1}-a^*\|\leq c\|F^*(a^{k+1})-F^{v,k}(a^{k+1})\|.
        \end{split}
    \end{equation}
    
Substituting \eqref{eq:F^k}, \eqref{eq:F^*}, and \eqref{eq:residual_error_disturb} into \eqref{eq:akv-a*}, we have
\begin{equation}\label{eq:37_1}
  \begin{split}
    &\|a^{k+1}-a^*\|\\
    &\leq c\|H(a^*,0)-H(a^k,v^k)\|\|a^{k+1}-a^k\|+\| \mathbf{e}_{F}(a^k,v^k)\|.
\end{split}  
\end{equation}
Furthermore, by substituting $
    \|a^{k+1}-a^k\|\leq \|a^{k+1}-a^*\|+\|a^{k}-a^*\|$
 into \eqref{eq:37_1} and by noting that $\|H(a^*,0)-H(a^k,v^k)\|$ and $\|\mathbf{e}_F(a^k,0)\|$ can be made arbitrarily small by shrinking $\delta_a$ and $\delta_v$ if necessary, there exist $L_a<1$ and $L_v<\infty$ such that
\begin{small}
\begin{equation}\label{eq:37}
  \begin{split}
    &\|a^{k+1}-a^*\|\\
    &\leq \frac{c\left(\|H(a^*,0)-H(a^k,v^k)\|\|a^{k}-a^*\|+ \|\mathbf{e}_F(a^k,v^k)\| \right)}{1-c\left[\|H(a^*,0)-H(a^k,v^k)\| \right]}\\
    &\leq \frac{c\left(\|H(a^*,0)-H(a^k,v^k)\|\|a^{k}-a^*\|+ \|\mathbf{e}_F(a^k,0)\|+L_{v}^F\|v^k\| \right)}{1-c\left[\|H(a^*,0)-H(a^k,v^k)\| \right]}\\
        &\leq L_a\|a^{k}-a^*\|+L_{v}\|v^k\|,
\end{split}  
\end{equation}
\end{small}
where $L_{v}^F$ is the locally Lipschitz modulus of $F$ with respect to $v$ at $(a^k,v^k)$. 


\subsection{Proof of Corollary \ref{C1}}
    According to \eqref{eq:37} and the Lipschitz continuity of $F$ and $H$ with respect to $(a,v)$,  there exist $L_H$ and $L'_v$ such that
 \begin{small}
\begin{equation}\label{eq:C1_1}
  \begin{split}
    &\|a^{k+1}-a^*\|\\
        &\leq c\left(\|H(a^*,0)-H(a^k,v^k)\|\|a^{k}-a^*\|+ \|\mathbf{e}_F(a^k,0)\|+L_{v}^F\|v^k\| \right)\\
        &\leq c\|\mathbf{e}_F(a^k,0)\|+cL_H \|a^{k}-a^*\|^2+L'_v\|v^k\|.
\end{split}  
\end{equation}
\end{small}
    Combining  \eqref{eq:q_qua} and \eqref{eq:C1_1}, there must exists a constant $L'_a$ such that \eqref{eq:ISS_perturbed_Newton_Q} holds.

\subsection{Proof of Theorem \ref{convergence_dis_local}}
 Eq. \eqref{eq:distributed_JN} is equivalent to solving VI$(\mathcal{A}_i,F^k_{i,{\Bar{a}_{-i}}})$, where $
    F^k_{i,{\Bar{a}_{-i}}}(a_i)=\nabla_{a_i}\Bar{J}_i(a_i^k,\Bar{a}_{-i})+\nabla^2_{a_ia_i}\Bar{J}_i(a_i^k,\Bar{a}_{-i})(a_i-a_i^k)$.
Denote $
    F^*_{i,{\Bar{a}_{-i}}}(a_i)=\nabla_{a_i}\Bar{J}_i(a_i^*,\Bar{a}_{-i})+\nabla^2_{a_ia_i}\Bar{J}_i(a_i^*,\Bar{a}_{-i})(a_i-a_i^*)$.
 
 Because $\nabla^2_{a_ia_i}\Bar{J}_i(\Bar{a}_i^*,\Bar{a}_{-i})$ is strictly copositive on the critical cone $\mathcal{C}_i(\Bar{a}_i^*;\mathcal{A}_i,F^*_{i,{\Bar{a}_{-i}}})$, $\Bar{a}_i^*$ is a stable solution to VI$(\mathcal{A}_i,F^*_{i,{\Bar{a}_{-i}}})$ \cite[Corollary 5.1.8]{book_pang}. The convergence of $\{a_i^k\}$ to $\Bar{a}_i^*$ then follows from a similar argument as in Theorem \ref{tm:stability_NE}. The convergence of the proximal response to a NE under the monotone $F$  follows from \cite[Theorem 16]{best_response}. 

\subsection{Proof of Theorem \ref{thm:ISS_semi}}
Denote $d^k=z^{k+1}-z^k$ and  $\mathbf{G}(z^k,d^k)=J\Phi(z^k)d^k$, $z^{k+1}$ generated from \eqref{eq:inexact_seminewton} satisfies
\begin{equation}
    \Phi(z^{k})+\mathbf{G}(z^k,d^k)=r^k.
\end{equation}

Because $\Phi$ in \eqref{system_equations} is strongly semismooth everywhere, according to \cite[Theorem 7.5.3]{book_pang}, $\text{Jac }\Phi(z)$ defines a nonsingular strong  Newton approximation scheme for $\Phi$ at $z^*$.

By the uniform Lipschitz homeomorphism property of $\mathbf{G}^{-1}(z,\cdot)$ with constant $L_{G^{-1}}$, we have
\begin{equation}\label{proof1_eq_1}
    \begin{split}
        \|z^{k+1}-z^*\|&=\|z^{k}+d^k-z^*\|\\
        &=\|z^{k}-z^*+\mathbf{G}^{-1}(z^k,-\Phi(z^k)+r^k)\|\\
        &=\|\mathbf{G}^{-1}(z^k,-\Phi(z^k)+r^k)\\
        &\quad-\mathbf{G}^{-1}(z^k,\mathbf{G}(z^k,z^*-z^{k}))\|\\
        &\leq L_{G^{-1}}\| \Phi(z^k)-r^k+\mathbf{G}(z^k,z^*-z^{k})\|\\
        &\leq L_{G^{-1}}(\|\Phi(z^k)+\mathbf{G}(z^k,z^*-z^{k})\|+\|r^k\|).
    \end{split}
\end{equation}
Because $\text{Jac }\Phi(z)$ is a nonsingular strong  Newton approximation scheme for $\Phi$ at $z^*$,
\begin{equation}\label{eq:quadratic}
        \frac{\|\Phi(z)+\mathbf{G}(z,z^*-z)\|}{\|z-z^*\|^2}\leq L_G.
    \end{equation}
Substituting \eqref{eq:quadratic} into the last inequality in \eqref{proof1_eq_1}, we have that
\begin{equation}
    \begin{split}
        \|z^{k+1}-z^*\|
        &\leq L_z\|z^k-z^*\|^2+L_r\|r^k\|,
    \end{split}
\end{equation}
holds for some constants $L_z$  and $L_r$.

\subsection{Proof of Theorem \ref{thm:bounded_error}}
    The proof is analogous to the proof of Lemma 1 in \cite{liao2020time} by noting that $\mathcal{T}$ from the developed Newton methods is at least $Q$-quadratically convergent and is locally ISS to disturbances. 

\bibliography{reference}

@article{GNEP,
  title={Generalized {Nash} equilibrium problems},
  author={Facchinei, Francisco and Kanzow, Christian},
  journal={Annals of Operations Research},
  volume={175},
  number={1},
  pages={177--211},
  year={2010},
  publisher={Springer}
}

@article{GNEP_KKT,
  title={On the solution of the {KKT} conditions of generalized {Nash} equilibrium problems},
  author={Dreves, Axel and Facchinei, Francisco and Kanzow, Christian and Sagratella, Simone},
  journal={SIAM Journal on Optimization},
  volume={21},
  number={3},
  pages={1082--1108},
  year={2011},
  publisher={SIAM}
}

@article{GPG_def,
  title={Decomposition algorithms for generalized potential games},
  author={Facchinei, Francisco and Piccialli, Veronica and Sciandrone, Marco},
  journal={Computational Optimization and Applications},
  volume={50},
  pages={237--262},
  year={2011},
  publisher={Springer}
}

@article{dominic_2,
  title={Online feedback equilibrium seeking},
  author={Belgioioso, Giuseppe and Liao-McPherson, Dominic and de Badyn, Mathias Hudoba and Bolognani, Saverio and Smith, Roy S and Lygeros, John and D{\"o}rfler, Florian},
  journal={IEEE Transactions on Automatic Control},
  year={2024},
  publisher={IEEE}
}

@article{Lacra_dissipativity,
  title={Dissipativity theory in game theory: On the role of dissipativity and passivity in {Nash} equilibrium seeking},
  author={Pavel, Lacra},
  journal={IEEE Control Systems Magazine},
  volume={42},
  number={3},
  pages={150--164},
  year={2022},
  publisher={IEEE}
}

@book{dontchev2009implicit,
  title={Implicit functions and solution mappings},
  author={Dontchev, Asen L and Rockafellar, R Tyrrell},
  volume={543},
  year={2009},
  publisher={Springer}
}

@article{NE_existence,
  title={On the existence of pure and mixed strategy {Nash} equilibria in discontinuous games},
  author={Reny, Philip J},
  journal={Econometrica},
  volume={67},
  number={5},
  pages={1029--1056},
  year={1999},
  publisher={Wiley Online Library}
}

@article{CGMPC_application2,
  title={Game theoretic model predictive control for distributed energy demand-side management},
  author={Stephens, Edward R and Smith, David B and Mahanti, Anirban},
  journal={IEEE Transactions on Smart Grid},
  volume={6},
  number={3},
  pages={1394--1402},
  year={2014},
  publisher={IEEE}
}

@article{GNEP_newton,
  title={Generalized {Nash} equilibrium problems and {Newton} methods},
  author={Facchinei, Francisco and Fischer, Andreas and Piccialli, Veronica},
  journal={Mathematical Programming},
  volume={117},
  number={1},
  pages={163--194},
  year={2009},
  publisher={Springer}
}

@article{Lacra_distributed2,
  title={Distributed {GNE} seeking under partial-decision information over networks via a doubly-augmented operator splitting approach},
  author={Pavel, Lacra},
  journal={IEEE Transactions on Automatic Control},
  volume={65},
  number={4},
  pages={1584--1597},
  year={2019},
  publisher={IEEE}
}

@book{dontchev2021lectures,
  title={Lectures on variational analysis},
  author={Dontchev, Asen L},
  volume={205},
  year={2021},
  publisher={Springer}
}

@article{pang_bookchapter,
  title={{Nash} equilibria: The variational approach},
  author={Facchinei, Francisco and Pang, Jong-Shi},
  journal={Convex optimization in signal processing and communications},
  pages={443},
  year={2010}
}

@book{book_pang,
  title={Finite-dimensional variational inequalities and complementarity problems},
  author={Facchinei, Francisco and Pang, Jong-Shi},
  year={2003},
  publisher={Springer}
}

@article{liao2020time,

  title={Time-distributed optimization for real-time model predictive control: Stability, robustness, and constraint satisfaction},

  author={Liao-McPherson, Dominic and Nicotra, Marco M and Kolmanovsky, Ilya},

  journal={Automatica},

  volume={117},

  pages={108973},

  year={2020},

  publisher={Pergamon}

}

@article{cunis2024input,
  title={Input-to-State Stability of {Newton} Methods for Generalized Equations in Nonlinear Optimization},
  author={Cunis, Torbj{\o}rn and Kolmanovsky, Ilya},
  journal={Proceedings of IEEE Conference on Decision and Control},
  year={2024}
}

@article{nash1950equilibrium,
  title={Equilibrium points in n-person games},
  author={{Nash} Jr, John F},
  journal={Proceedings of the National Academy of Sciences},
  volume={36},
  number={1},
  pages={48--49},
  year={1950},
  publisher={National Acad Sciences}
}

@phdthesis{josephy_1,
  title={Quasi-{Newton} methods for generalized equations},
  author={Josephy, Norman},
  year={1979},
  school={University of Wisconsin}
}

@article{dontchev2013convergence,
  title={Convergence of inexact {Newton} methods for generalized equations},
  author={Dontchev, Asen L and Rockafellar, R Tyrrell},
  journal={Mathematical Programming},
  volume={139},
  number={1},
  pages={115--137},
  year={2013},
  publisher={Springer}
}

@article{quasirefular_vs_stronglyregular,
  title={On the accurate identification of active constraints},
  author={Facchinei, Francisco and Fischer, Andreas and Kanzow, Christian},
  journal={SIAM Journal on Optimization},
  volume={9},
  number={1},
  pages={14--32},
  year={1998},
  publisher={}
}

@article{best_response,
  title={Real and complex monotone communication games},
  author={Scutari, Gesualdo and Facchinei, Francisco and Pang, Jong-Shi and Palomar, Daniel P},
  journal={IEEE Transactions on Information Theory},
  volume={60},
  number={7},
  pages={4197--4231},
  year={2014},
  publisher={IEEE}
}
\bibliographystyle{IEEEtran}

\end{document}